\documentclass[11pt]{article}  
\usepackage[margin=2.54cm]{geometry} 
\usepackage{hyperref}
\usepackage{bm}
\usepackage{amsfonts}
\usepackage{amsmath}
\usepackage{amssymb}
\usepackage{amsthm}
\usepackage{graphicx}
\usepackage{xspace}

\newcommand{\pr}[2]{\ensuremath{\textit{\bf Pr}_{#1}\left[ #2 \right]}\xspace}
\newcommand{\e}[2]{\ensuremath{\textit{\bf E}_{#1}\left[ #2 \right]}\xspace}
\newcommand{\sd}[2]{\ensuremath{\bm{\sigma}_{#1}\left[ #2 \right]}\xspace}

\newtheorem{theorem}{Theorem}
\newtheorem{definition}{Definition}
\newtheorem{corollary}{Corollary}
\newtheorem{lemma}{Lemma}
\newtheorem{proposition}{Proposition}

\begin{document}

\title{Bayesian Incentive Compatibility via Fractional Assignments}
\author{Xiaohui Bei~\thanks{Institute for Theoretical Computer Science,
Tsinghua University. Email: {\tt bxh08@mails.tsinghua.edu.cn}.
Supported in part by the National Natural Science Foundation of 
China Grant 60553001, the National Basic Research Program of China Grant 
2007CB807900, 2007CB807901.}\\
\and 
Zhiyi Huang~\thanks{Computer and Information Science,
University of Pennsylvania. Email: {\tt hzhiyi@cis.upenn.edu}.
This material is based upon work supported by the National Science Foundation 
under Grant No. 0716172}}
\date{}

\begin{titlepage}
	\thispagestyle{empty}
	\maketitle

\begin{abstract}
	\thispagestyle{empty}
Very recently, Hartline and Lucier \cite{hartline2009bayesian} studied
single-parameter mechanism design problems in the Bayesian setting. They
proposed a black-box reduction that converted Bayesian approximation 
algorithms into Bayesian incentive compatible (BIC) mechanisms and preserved
the expected social welfare. It remains a open question if one can find 
similar reduction in the more important multi-parameter setting.
In this paper, we give positive answer to this question
when the prior distribution has finite and small support.
We propose a black-box reduction for designing BIC multi-parameter
mechanisms. The reduction converts any algorithm into an
$\epsilon$-BIC mechanism with only marginal loss in social welfare. 
As a result, for combinatorial auctions with sub-additive agents
we get an $\epsilon$-BIC mechanism that achieves constant approximation.
\end{abstract}
\end{titlepage}

\section{Introduction}

In this paper, we consider the problem of designing computationally efficient
and truthful mechanism for multi-parameter mechanism design problems in
the Bayesian setting. 

Suppose a major Internet search service provider wants to sell multiple
advertisement slots to a number of companies. From the history of previous 
transactions, we can estimate a prior distribution of each company's
valuation of the advertisement slots. What mechanism shall
the search service provider use to obtain good social welfare, or good revenue?

This is a typical multi-parameter mechanism design problem.
In general, we consider the scenario in which a principal wants to sell 
a number of different services to multiple 
heterogeneous strategic agents subject to some feasibility constraints (e.g.
total cost of providing these services must not exceed the budget), so
that some desired objective (e.g. social welfare, revenue, residual surplus) 
is achieved. If we interpret this as simply a combinatorial optimization
problem, then there exists approximation algorithms for many of these problems.
And the approximation ratios of many of these algorithms are tight subject
to certain computational complexity assumptions. However, if we wants to 
design protocols of allocations and setting prices in order to
achieve the desired objective in the equilibrium strategic behavior of the 
agents, we usually have much worse approximation ratio.
Therefore, it is natural to ask the following question:
\begin{quote}
	{\em Can we convert any algorithm into a truthful mechanism while 
	preserving the performance, say, social welfare?}
\end{quote}

Unfortunately, from previous work we learn that this is impossible for
some problems. Papadimitriou et al. 
\cite{papadimitriou2008hardness} showed the first significant gap between
the performance of deterministic algorithms and deterministic truthful
mechanisms via the Combinatorial Public Project problem.

\paragraph{Bayesian setting.}

The standard game theoretic model for incomplete information is the 
Bayesian setting, in which the agent valuations 
are drawn from a publicly 
known distribution. The standard solution concept in this setting is 
{\em Bayesian-Nash Equilibrium}. In a Bayesian-Nash equilibrium, each player 
maximizes its expected payoff by following the strategy profile
given the prior distribution of the agent valuations.

\medskip

In this paper, we will consider multi-parameter welfare-preserving
algorithm/mechanism reductions 
in the Bayesian setting, and weaken truthfulness
constraint from Incentive Compatibility (IC) to Bayesian Incentive 
Compatibility (BIC), which means truth telling is the equilibrium strategy
over random choice of the mechanism as well as the random realization of 
the other agent valuations. In many real world applications
such as online auctions, AdWords auctions, spectrum auctions etc., the
availability of data of past transactions make it possible to obtain
good estimation of the prior distribution of the agent valuations.
Thus, revisiting the algorithm/mechanism reduction problem 
in the Bayesian setting is of both theoretical and practical importance.

Hartline and Lucier \cite{hartline2009bayesian} studied 
this problem in the single-parameter
setting. They showed a brilliant black-box reduction from any approximation 
algorithm
to BIC mechanism that preserves the performance with respect to social welfare
maximization. In this paper, we prove that similar reduction also exists for
the realm of multi-parameter mechanism design for social welfare! Moreover,
we can also obtain BIC mechanism for revenue or residual surplus via some
variants of our black-box reduction.

\paragraph{Our results and technique.} 
Our main result is a black-box reduction 
that converts algorithms into BIC mechanisms with essentially 
the same social welfare for arbitrary multi-parameter mechanism design problem
in the Bayesian setting. More concretely, given an 
algorithm $\mathcal{A}$ that provides $SW^\mathcal{A}$ social welfare,
the reduction provides a mechanism that gives 
$SW^\mathcal{A} - \epsilon$ social welfare and is $\epsilon$-BIC.
The running time is polynomial in the input size and $1/\epsilon$. 
This resolves an open problem in \cite{hartline2009bayesian}. The key idea is 
to decouple the reported valuations and the input valuations for the algorithm $\mathcal{A}$. 
When the reported valuations are $v_1, v_2, \dots, 
v_n$, we will manipulate the valuations via some carefully designed 
intermediate algorithms $\mathcal{B}_1, \dots, \mathcal{B}_n$, and use allocation
$\mathcal{A}(\mathcal{B}_1(v_1), \dots, \mathcal{B}_n(v_n))$. 
We prove that there exist intermediate algorithms
$\mathcal{B}_1, \dots, \mathcal{B}_n$ so that there are prices that achieve BIC.
Under certain conditions, the marginal loss factor in social welfare can be made multiplicative.

As an application of this reduction, we get a $(\frac{1}{2} - \epsilon)$-approximate and 
$\epsilon v_{max}$-BIC mechanism for social welfare maximization in 
combinatorial auctions with sub-additive agents. 
For the more restricted case of fractionally
sub-additive agents, we obtain $(1 - \frac{1}{e} - \epsilon)$-approximate
mechanism. 

\paragraph{Related work.}
The problem of maximizing social welfare against strategic agents is one of
the oldest and most famous problems in the area of mechanism design. It has
been extensively studied by the economists in both Bayesian and prior-free
setting without considering computational power constraint. The celebrated
VCG mechanism \cite{clarke1971multipart, groves1973incentives,
vickrey1961counterspeculation} which guarantees optimal social welfare and
incentive compatibility is one of the most exciting results 
in this domain. However, implementing the VCG mechanism is NP-hard
in general. This is one of the reasons that VCG mechanism is rarely used in
practice despite of its lovely theoretical features.

In the past decade, computer scientists introduced many novel techniques in
the prior-free setting to design computationally efficient mechanisms
that provide incentive compatibility and/or good approximation to
optimal social welfare for various families of valuation functions.

On the one hand, 
Dobzinski, Nisan and Schapira \cite{dobzinski2006truthful} proposed
poly-time mechanisms which achieved $\Omega(1 / \sqrt{n})$-approximation for
general agents and $\Omega(1 / \log^2 n)$-approximation for sub-modular 
agnets. Dobzinski \cite{dobzinski-two} later proposed a truthful mechanism 
which achieved an improved
$\widetilde{\Omega}(1 / \log n)$-approximation for a strictly broader class of
sub-additive agents.

On the other hand, 
if we focus on the algorithmic problem of maximizing social welfare
assuming all valuations are truthfully revealed, then the algorithm by
Dobzinski, Nisan and Schapira \cite{dobzinski2005approximation} 
gave $\Omega(1 / \sqrt{n})$-approximation for
general case and $\Omega(1 / \log n)$-approximation for sub-additive agents.
The latter approximation ratio is later improved to $\frac{1}{2}$
for sub-additive agents \cite{feige2006maximizing} and $(1 - \frac{1}{e})$ for
the more restricted  class of fractionally sub-additive agents 
\cite{dobzinski-two,feige2006approximation}.

The above results suggest that there exists a gap between the performance of
the best poly-time algorithms and that of the best poly-time and incentive 
compatible mechanism.
As an effort to study the relation between designing algorithms and
designing truthful mechanisms with good approximation ratio, Lavi and Swamy
\cite{lavi2005truthful} proposed a meta-mechanism that converted
strong rounding algorithms for the standard LP of social welfare maximization
into IC mechanisms. However,
their approach required the rounding algorithm to work for arbitrary valuation
functions. This requirement prevents their technique to get good
approximation beyond cases of general valuations and additive valuations 
(via a different linear program). But the more interesting classes of 
valuations (e.g. sub-additive valuations and sub-modular valuations) lies between these two
extremes. Another notable attempt on reducing IC mechanism design to algorithm design
is the very recent work by Dughmi and Roughgarden~\cite{dughmi2010black}. They proved that
for any packing problem that admitted an FPTAS, there was an IC mechanism that was also an FPTAS.

Most of the previous effort from computer scientists has focused on the
prior-free setting. Until very recently, there has been a few work
that brought more and more Bayesian analysis into the field of algorithmic
mechanism design. Hartline
and Lucier \cite{hartline2009bayesian} gave a black-box reduction that
converted any Bayesian approximation algorithm into a Bayesian incentive
compatible mechanism that preserved 
social welfare in the single parameter domain.
Bhattacharya et al. \cite{bhattacharya2009budget} studied the revenue
maximization problem for auctioning heterogeneous items when the valuations
of the agents were additive. Their result gave constant approximation in
the Bayesian setting even when the agents had public known budget constraints.
Chawla et al. \cite{chawla2009multi} considered the revenue maximization 
problem in the multi-dimensional multi-unit auctions. 
They introduced mechanism that
gave constant approximation in various settings via sequential posted pricing.

Finally, in concurrent and independent work, Hartline et al. \cite{jason2010} 
study the relation of algorithm and mechanism in Bayesian setting and propose
similar reduction. In the discrete support setting that is considered 
in this paper, they use essentially the same reduction. However,
their work achieves perfectly BIC instead of $\epsilon$-BIC. They also extend
the reduction to the more general continuous support setting.

\section{Preliminaries}

\subsection{Notations.}

We use $\{x_i\}_{1 \leq i \leq n}$ to denote an array of 
size $n$. We also use the natural extension of this notation for 
multi-dimensional arrays. 
We will use bold font $\bm{x}$ to denote a vector $(x_1, \dots, x_n)$.
We let $\Delta(S)$ denote the set of distributions over the elements in a set 
$S$. For a random variable $x$, we let $\e{}{x}$
denote its expectation and let $\sd{}{x}$ denote its standard deviation. 
We use subscripts to represent the random choices over which we consider
the expectation and variance. For instance, $\e{y \sim F}{x}$
is the expectation of $x$ when $y$ is drawn from distribution
$F$. We sometimes use $\e{y}{x}$ for short when the distribution $F$ is 
clear from the context.

\subsection{Model and definitions.}

In this section, we will formally introduce the model in this 
paper. We study the general multi-parameter mechanism design problems. In a
multi-parameter mechanism design problem, a principal wants to sell a set of 
services to multiple heterogeneous agents in order to optimize the desired 
objective (e.g. social welfare, revenue, residual surplus, etc.). 
A Bayesian multi-parameter mechanism design problem with $n$ agents is defined 
by a tuple $\langle \bm{I}, \bm{J}, \bm{V}, \bm{F} \rangle$. 

\begin{itemize}
	\item $\bm{I} = (I_1, \dots, I_n)$:  
		The set of services that the principal wants to sell to the agents. 

		Since we can impose arbitrary feasibility constraints on the 
		allocations, 
		we assume without loss of generality that the services are 
		partitioned into $n$ disjoint sets $I_1, \dots, I_n$ so that the 
		services in $I_i$ only aim for agent $i$, and each agent $i$ is 
		interested in any one of these services.

	\item $\bm{J} \subseteq I_1 \times \dots \times I_n$: 
		The set of feasible allocations.
		
	\item $\bm{V} = V_1 \times \dots \times V_n$:
		The space of agent valuations.
		
		We let $V_i \subseteq \mathbb{R}^{I_i}$ denote
		the set of possible valuations of agent $i$. 
		We let $v_{max}$ denote the maximal valuation, that is,
		$v_{max} = \max_{i, v \in V_i, S \in I_i} v(S)$ 
	\item $\bm{F} = F_1 \times F_2 \times \dots \times F_n$: The
		joint prior distribution of the agent valuations.
		
		We assume the prior distribution is a product distribution.
		We let $F_i \in \Delta(V_i)$ denote the prior distribution
		of the valuation of agent $i$. In this paper, 
		we only consider distributions with finite and polynomially large
		support. We will assume without 
		loss of generality that the support of each distribution $F_i$ is 
		$\{v^1_i, \dots, v^\ell_i\}$. Suppose $v_i \sim F_i$, We will let
		$f_i(t)$ denote the probability that $v_i = v^t_i$.
\end{itemize}

For example, in the {\em combinatorial auction}
problem with $n$ agents and $m$ items, we let $[m] = \{1, 2, \dots, m\}$
denote the set of items. The set of services for each agent $i$ is
the set of all subsets of items, that is, $I_i = 2^{[m]}$, $1 \leq i \leq n$.
The set of feasible allocations is
$$\bm{J} = \{(S_1, \dots, S_n) : S_i \in I_i, S_i \cap S_j = \emptyset\}
\enspace.$$

The set of valuations, $V_i$, is the set of mappings from subset of items
$I_i$ to $\mathbb{R}_+$ that are monotone ($v_i(S) \leq v_i(T)$ for $S \subseteq T$)
and normalized ($v_i(\emptyset) = 0$).
We usually assume that the valuations in $\bigcup_i V_i$ satisfies certain 
properties, e.g. sub-additivity, sub-modularity, etc. 

\paragraph{Algorithm.} An algorithm for a multi-parameter mechanism design 
problem $\langle \bm{I}, \bm{J}, \bm{V}, \bm{F}\rangle$ is a protocol (may
or may not be randomized) that
takes a realization of agent valuations $\bm{v} \in \bm{V}$ as input, and
outputs a feasible allocation $\bm{S} \in \bm{J}$. 

\paragraph{Mechanism.} A mechanism is an interactive protocol (may or may not
be randomized) between the 
principal and the agents so that the principal can retrieve information from
the agents (presumably via their bids), and determine an allocation of 
services $\bm{S} \in \bm{J}$ and a collection of prices 
$\bm{p} = (p_1, \dots, p_n)$. The extra challenge for mechanism design,
compared to algorithm design, is to retrieve genuine valuations from the
agents and handle their strategic behavior.

For each $1 \leq i \leq n$, we will assume the prior distribution $F_i$ is 
{public known}. But the actual realization $v_i \sim F_i$ is {\em private}
information of agent $i$. Each agent $i$ aims to maximizes the quasi-linear
utility $v_i(S_i) - p_i$, where $S_i$ is the service it gets and $p_i$ is
the price. Thus, the agents may not reveal their genuine valuations if
manipulating their bids strategically can increase their utility.

\paragraph{Objectives.}
We will consider
three different objectives: social welfare, revenue, and residual surplus.
The expected {\em social welfare} of a mechanism $\mathcal{M}$ is
$$SW^\mathcal{M} = \e{\bm{v} \sim \bm{F}, (\bm{S}, \bm{p}) \sim \mathcal{M}(\bm{v})}
{\sum_{i = 1}^n v_i(S_i)} \enspace.$$

Similarly, we will let $SW^\mathcal{A}$ denote the expected
social welfare of an algorithm $\mathcal{A}$.

\begin{definition}
	An algorithm $\mathcal{A}$ is $\alpha$-approximate in social welfare for 
	a multi-parameter mechanism design problem 
	$\langle \bm{I}, \bm{J}, \bm{V}, \bm{F} \rangle$, if 
	$SW^\mathcal{A} \geq \alpha \, \mathsf{OPT}$.
\end{definition}

The expected {\em revenue} of a mechanism is
$$R^\mathcal{M} = \e{\bm{v} \sim \bm{F}, (\bm{S}, \bm{p}) \sim \mathcal{M}(\bm{v})}
{\sum_{i = 1}^n p_i} \enspace.$$

The last objective, {\em residual surplus}, was recently proposed by Hartline
and Roughgarden \cite{hartline2008optimal} as an alternative objective in
the flavour of social welfare. In the residual surplus maximization problem, 
the principal aims to maximize the sum of the agents' utilities instead of 
the sum of their valuations. The expected {\em residual surplus} is
$$RS^\mathcal{M} = \e{\bm{v} \sim \bm{F}, (\bm{S}, \bm{p}) \sim \mathcal{M}(\bm{v})}
{\sum_{i = 1}^n \left(v_i(S) - p_i\right)} \enspace.$$

We will let $\mathsf{OPT}$ denote the optimal social welfare, that is, $\mathsf{OPT}
= \max_\mathcal{M} SW^\mathcal{M}$. Since both revenue and residual surplus are 
upper-bounded by social welfare. We will use $\mathsf{OPT}$ as our benchmark 
for all three objectives.

\paragraph{Solution concepts.}

Ideally, a mechanism shall provide incentive for the agents to reveal their
valuations truthfully. In this section, we will formalize this requirement
by introducing the game-theoretical solution concepts that we use in this 
paper.

\begin{definition}
	A mechanism $\mathcal{M}$ is {\em Bayesian incentive compatible (BIC)} if
	for each agent $i$ and any two valuations $v_i, \widetilde{v}_i \in V_i$, 
	we have
	$$\e{\bm{v}_{-i},(\bm{S},\bm{p}) \sim \mathcal{M}(v_i,\bm{v}_{-i})}
	{v_i(S_i)-p_i} \geq \e{\bm{v}_{-i},(\bm{S},\bm{p}) \sim 
	\mathcal{M}(\widetilde{v}_i, \bm{v}_{-i})} {v_i(S_i)-p_i} \enspace.$$
\end{definition}

\begin{definition}
	A mechanism $\mathcal{M}$ is {\em $\epsilon$-Bayesian Incentive Compatible
	($\epsilon$-BIC)} if for any agent $i$ and any two valuations
	$v_i, \widetilde{v}_i \in V_i$,
	$$\e{\bm{v}_{-i}, (\bm{S}, \bm{p}) \sim 
	\mathcal{M}(v_i, \bm{v}_{-i})}{v_i(S_i) - p_i} \geq 
	\e{\bm{v}_{-i}, (\bm{S}, \bm{p}) \sim 
	\mathcal{M}(\widetilde{v}_i, \bm{v}_{-i})}{v_i(S_i) - p_i} 
	- \epsilon \enspace.$$
\end{definition}

Other than the above constraints of incentive compatibility, the mechanism 
shall also guarantee that the agents always get non-negative utility.
Otherwise, the agents may choose not to participate in the 
mechanism. This is known as the {\em individual rationality} constraint.

\begin{definition}
	A mechanism $\mathcal{M}$ is {\em individually rational (IR)}
	if for any realization $\bm{v}$ of agent valuations, and any allocation
	$\bm{S}$ and prices $\bm{p}$ by the mechanism, we always have that 
	$v_i(S_i) - p_i \geq 0$ for all agent $i$,
\end{definition}

\section{Characterization of BIC mechanisms}
\label{sec:charbic}

In this section, we will introduce a 
non-trivial characterization of BIC multi-parameter mechanisms via a novel
connection between BIC mechanisms and envy-free prices. This characterization
inspires our reduction in the next section. 

\subsection{Fractional assignment problem.}

We will first introduce the fractional assignment problems,
which will play a critical role in the results of this paper,
and a useful lemma about envy-free prices in fractional assignment problems.

In order to distinguish the notations for fractional assignment problems and 
those for the mechanism design problems, we will use superscripts instead
of subscripts to specify different entries of a vector for the fractional
assignment problems. For instance, we will use
$x^s$ to denote the $s^{th}$ entry of a vector $\bm{x}$.

Let us consider a market with $\ell$ buyers and $m$ infinitely divisible 
products. 
Each buyer $s$ has a non-negative demand $\alpha^s$, which denotes the maximal 
amount of products the buyer will buy. Each product $t$ has a non-negative 
supply $\beta^t$, which denotes the available amount of this product in the 
market. For each buyer $s$ and each product $t$, we let $w^{st}$ denote the
non-negative value of buyer $s$ of product $t$. 

The goal is to set prices for the products and to assign the products to the 
buyers subject to the demand and supply constraints. Thus, a solution 
$(\bm{x}, \bm{p})$ to the fractional assignment problem 
consists of a collection of prices $\bm{p} = (p^1, \dots, p^\ell)$ 
and a feasible allocation $\bm{x} = \{x^{st}\}_{1 \leq s \leq \ell, 1 \leq t 
\leq m}$ in the polytope:
$$\left\{\bm{x} : \forall s, \sum_{t = 1}^m x^{st} \leq \alpha^s; \,
\forall t, \sum_{s = 1}^\ell x^{st} \leq \beta^t; \,
\bm{x} \geq 0\right\} \enspace,$$
where $x^{st}$ denotes the amount of product $t$ that is 
assigned to buyer $s$.

\begin{definition}
	A solution $(\bm{x}, \bm{p})$ is {\em envy-free} if for any $x^{st} > 0$, 
	then $t$ is a product that maximizes the quasi-linear utility of agent 
	$s$, and the utility for agent $s$ is non-negative. That is, 
	\begin{equation}
		\label{eq:envyfree}
		\forall s, t : x^{st} > 0 \Rightarrow 
		w^{st} - p^t = \max_k \{w^{sk} - p^k\} \geq 0 \enspace.
	\end{equation}
\end{definition}

\begin{definition}
	An allocation $\bm{x}$ is {\em market-clearing}
	if all demand constraints and supply constraints hold with equality
	That is, 
	$$\forall 1 \leq s \leq \ell: \sum_{t=1}^m x^{st} = \alpha^s ~,~
	\forall 1 \leq t \leq m: \sum_{s=1}^\ell x^{st} = \beta^t \enspace.$$
\end{definition}

The social welfare maximization problem for a fractional assignment problem
is characterized by the following linear program (P) and its dual (D).
\begin{align*}
	\text{{\bf (P)} Maximize} & ~ ~ \Sigma_{s = 1}^\ell 
	\Sigma_{t = 1}^m x^{st} w^{st} & & \text{s.t.} &
	\text{{\bf (D)} Minimize} & ~ ~ \Sigma_{s = 1}^\ell \alpha^s u^s + 
	\Sigma_{t = 1}^m \beta^t p^t & & \text{s.t.} \\
	\Sigma_{t = 1}^m x^{st} & \leq \alpha^s & & \forall s & 
	u^s + p^t & \geq w^{st} & & \forall s, t \\
	\Sigma_{s = 1}^\ell x^{st} & \leq \beta^t & & \forall t &
	u^s & \geq 0 & & \forall s \\
	x^{st} & \geq 0 & & \forall s, t &
	p^t & \geq 0 & & \forall t
\end{align*}

\begin{align*}
\end{align*}

We will introduce two useful lemmas about the connection between 
envy-free prices and social welfare maximization for fractional
assignment problems. These lemmas were known in different forms in the
economics literature \cite{gul1999walrasian}.

\begin{lemma}
	\label{lemma:assignment1}
	If there exist envy-free prices $\bm{p}$ for a market-clearing
	allocation $\bm{x}$, then $\bm{x}$ maximizes the social welfare,
	that is, $\bm{x} \cdot \bm{w} = \max_{\bm{z}} \bm{z} \cdot \bm{w}$.
\end{lemma}

\begin{proof}
	Suppose there exist envy-free prices $\bm{p}$ for 
	an allocation $\bm{x}$. Let $u^s = \max_t \left\{w^{st}-p^t\right\}$. 
	We have that $u^s + p^t \geq w^{st}$ for all $s, t$. So 
	$(\bm{u}, \bm{p})$ is a feasible solution for the dual LP.

	Moreover, by definition of envy-freeness, we have
	$$\forall s, t : x^{st} > 0 \Rightarrow u^s = w^{st} - p^t \enspace.$$
	
	Therefore, we get that
	$$\sum_{s=1}^\ell \sum_{t=1}^m x^{st} w^{st} = 
	\sum_{s=1}^\ell \sum_{t=1}^m x^{st} (u^s + p^t) 
	= \sum_{s=1}^\ell \alpha^s u^s + \sum_t \beta^t p^t \enspace.$$

	The last equality holds because $\bm{x}$ is market clearing.
	Notice that $\bm{x}$ is a feasible solution to the primal LP. 
	By duality theorem, we get that the allocation $\bm{x}$ maximizes 
	the social welfare for the fractional assignment problem.
\end{proof}
 
\begin{lemma}
	\label{lemma:assignment2}
	If an allocation $\bm{x}$ maximizes the social welfare, then there exist
	envy-free prices $\bm{p}$ for the fractional assignment problem.
\end{lemma}

\begin{proof}
	Suppose the allocation $\bm{x}$ maximizes the social
	welfare. Let $(\bm{u}, \bm{p})$ be an optimal solution to the dual
	LP. By complementary slackness we get that $x^{st} > 0$ only if the
	corresponding dual constraint is tight, that is, $u^s + p^t = w^{st}$.
	Therefore, $x^{st} > 0$ implies that $w^{st} - p^t = u^s \geq 
	w^{sk} - p^k$ for all $k$. Thus $\bm{p}$ is a collection of envy-free 
	prices for the allocation $\bm{x}$ in this fractional assignment problem.
\end{proof}

Note that the above proof also gives a poly-time algorithm for
finding the welfare maximizing allocation $\bm{x}$ and the corresponding
envy-free prices $\bm{p}$ by solving the primal and dual LPs.
Moreover, we also get that the envy-free prices $\bm{p}$ satisfy a weak 
uniqueness in the sense that it must be part of an optimal solution for the 
dual LP.

\begin{corollary}
	There exists a poly-time algorithm that computes 
	the welfare-maximizing market-clearing allocation and 
	the envy-free prices. 
\end{corollary}

\subsection{Characterizing BIC via envy-free prices.}

We first introduce some notations that will simplify our discussion.
Given a mechanism $\mathcal{M}$ for a multi-parameter mechanism design problem
$\langle \bm{I}, \bm{J}, \bm{V}, \bm{F} \rangle$, we will consider the
expected values and expected prices for each agent when it choose a specific 
strategy (each strategy corresponds to reporting a specific valuation).

Assuming the other agents report their valuations truthfully, 
agent $i$'s expected value of the service it gets, when the 
genuine valuation is $v^s_i$ and the reported valuation is $v^t_i$, is
$$w^{st}_i = \e{\bm{v}_{-i},(\bm{S},\bm{p})\sim\mathcal{M}(v^t_i,\bm{v}_{-i})}{
v^s_i(S_i)} \enspace.$$

Similarly, we let $p^i_t$ denote the expected price the mechanism would charge
to agent $i$ if its reported valuation is $v^t_i$, that is, 
$$p^t_i = \e{\bm{v}_{-i},(\bm{S},\bm{p})\sim\mathcal{M}(v^t_i,\bm{v}_{-i})}{p_i} 
\enspace.$$

By the definition of BIC and IR, the mechanism $\mathcal{M}$ is BIC and IR if and 
only if for any $1 \leq i \leq n$ and $1 \leq s \leq \ell$,
\begin{equation}
	w^{ss}_i - p^s_i = \max_t \{w^{st}_i - p^t_i \} \geq 0
	\enspace. \label{eq:bic}
\end{equation}

The above equation \eqref{eq:bic} is similar to 
equation \eqref{eq:envyfree} in the definition of envy-freeness in fractional
assignment problem. In fact,
the key observation is that the above BIC condition is equivalent to the 
envy-free condition for a set of properly chosen fractional assignment 
problems.

\paragraph{Induced assignment problems.}

For each agent $i$, we will consider the following {\em induced assignment 
problem}. 
We consider $\ell$ virtual buyers with demands $f_i(1), \dots, f_i(\ell)$ 
respectively, and $\ell$ virtual products with supplies $f_i(1), \dots, 
f_i(\ell)$ respectively. For each virtual buyer $s$ and each virtual product
$t$, let virtual buyer $s$ has value $w^{st}_i$ on virtual product $t$.
We will refer to this fractional assignment problem the {\em induced
assignment problem} of agent $i$.

\medskip

Let us consider the {\em identity allocation} $\bm{x}_i$ defined as follows:
$$x^{st}_i = \left\{\begin{aligned}
	& f_i(s) & & \text{, if } s = t \enspace,\\
	& 0 & & \text{, otherwise.}
\end{aligned}\right.$$ 

We can easily verify that 
a collection of prices $\bm{p}_i = (p^1_i, \dots, p^\ell_i)$ satisfies
constraint \eqref{eq:bic} if and only if $\bm{p}_i$ satisfies the envy-free
condition \eqref{eq:envyfree} of the induced assignment problem of agent $i$
with respect to the above identity allocation.
Hence, we have the following connection between BIC mechanism and the
envy-free prices of the induced assignment problems.

\begin{lemma}[Characterization Lemma~\cite{rochet1987necessary}]
	\label{lemma:characterization}
	A mechanism $\mathcal{M}$ is BIC if and only if in the induced assignment 
	problem of each agent $i$ the identity 
	allocation $\bm{x}_i = \{x^{st}_i\}_{1 \leq s,t \leq \ell}$ 
	maximizes the social welfare, and $\bm{p}_i = (p^1_i, \dots, p^\ell_i)$ 
	are chosen to be the corresponding envy-free prices.
\end{lemma}

\paragraph{Comparing with Myerson's characterization.}

Suppose the problem falls into the single-parameter domain. Each
valuation $v^s_i$ is represented by a single non-negative real number. With
a little abuse of notation, we let $v^s_i$ denote this value. Without loss
of generality, we assume that $v^1_i > \dots > v^\ell_i$. We let
$y^t_i$ denote the probability that agent $i$ would be served if the reported
value was $v^t_i$. The values $\bm{w}_i$ in the fractional assignment
problems of agent $i$ are $w^{st}_i = v^s_i y^t_i$ for $1 \leq s, t \leq \ell$.
Myerson's famous characterization \cite{myerson1981optimal}
of truthfulness in single-parameter domain
implied that the mechanism is BIC if and only if $y^1_i \geq
\dots \geq y^\ell_i$. Indeed, due to rearrangement inequality, the identity 
allocation $\bm{x}_i$ maximizes the social welfare if and only if 
$y^1_i \geq \dots \geq y^\ell_i$. Thus, the characterization lemma implies
Myerson's characterization in the single-parameter domain.

\section{Reduction for social welfare}
\label{sec:socialwelfare}

Lemma \ref{lemma:characterization} suggests an interesting connection between 
BIC and envy-free prices for the induced assignment problems. 
Hence, given an algorithm $\mathcal{A}$, we will manipulate the allocation by 
$\mathcal{A}$ based on this connection in order to make 
it satisfy the condition in Lemma \ref{lemma:characterization}.

\subsection{Main ideas.}

Let us first briefly convey two key ideas in the construction of the
welfare-preserving reduction.

The first idea is to decouple the reported agent valuations and the 
input agent valuations for algorithm $\mathcal{A}$. More
concretely, we will introduce $n$ intermediate algorithm $\mathcal{B}_1, \dots,
\mathcal{B}_n$. Each $\mathcal{B}_i$ will take the reported valuation $v'_i$ as
input, then output a valuation $\widetilde{v}_i$ that may or may not equals 
$v'_i$. Then, we will use algorithm $\mathcal{A}$ to compute the allocation $\bm{S}$
for agent valuations $\widetilde{v}_1, \dots, \widetilde{v}_n$, 
and allocate services according to $\bm{S}$. 

If we revisit the values 
$\widetilde{\bm{w}}_i$ in the induced assignment problem of agent $i$ after
this manipulation, we will get that for any $1 \leq s, t \leq \ell$,
$$\widetilde{w}^{st}_i = \e{\bm{v}_{-i}, \widetilde{\bm{v}} \sim 
\mathcal{B}(v^t_i, \bm{v}_{-i}), \bm{S} \sim \mathcal{A}(\widetilde{\bm{v}})}
{v^s_i(S_i)} \enspace.$$

\begin{figure}
	\centering
	\includegraphics[width=.45\textwidth]{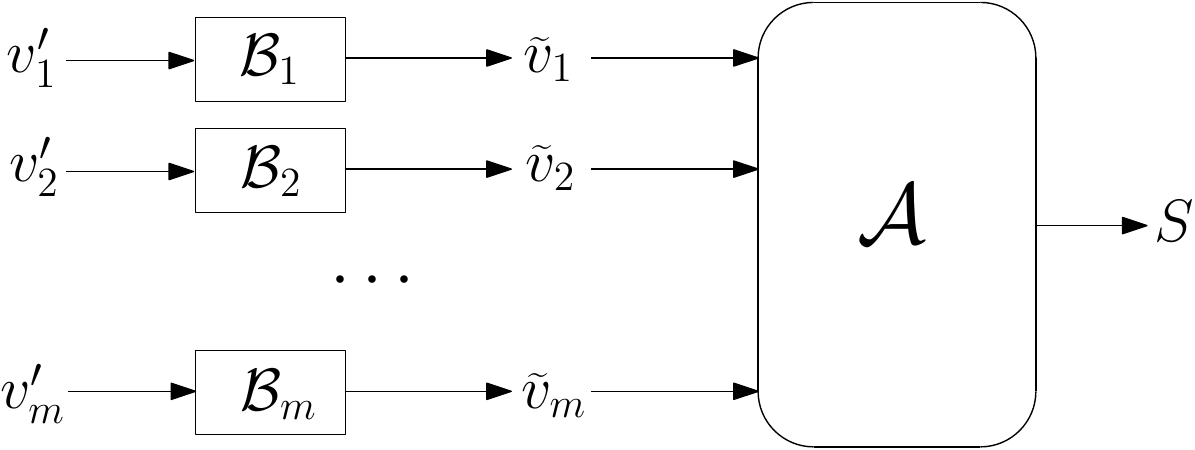}
	\caption{High-level picture of the reduction for social welfare.
	$\mathcal{B}_i$'s are intermediate algorithms for manipulating the 
	input of algorithm $\mathcal{A}$.
	$\widetilde{v}_i$'s are the reported valuations. $v'_i$'s are the 
	manipulated input valuations for algorithm $\mathcal{A}$. $\bm{S}$ is
	the final allocation.}
	\label{fig:reduction}
\end{figure}

By Lemma \ref{lemma:characterization}, we need to choose $\mathcal{B}_i$'s carefully, 
so that the identity allocations in the manipulated assignment problems are
welfare-maximizing allocations. However, from the above equation we can see 
that by using $\mathcal{B}_i$ to manipulate the $i^{th}$ valuation, we may change
not only the structure of the induced assignment problem of agent $i$, but the
structure of the induced assignment problems of other agents as well. Hence,
we need to handle such correlation among the induced assignment problems when
we choose intermediate algorithms $\mathcal{B}_1, \dots, \mathcal{B}_n$.

The idea that handles this correlation is to impose an extra constraint on
each intermediate algorithm $\mathcal{B}_i$: if the input valuation $v'_i$ is 
drawn from the distribution $F_i$, then the output valuation 
$\widetilde{v}_i$ also follows the same distribution, that is, 
for all $1 \leq i \leq n$ and $1 \leq t \leq \ell$,
\begin{equation}
	\label{eq:interconstraint}
	\pr{v'_i \sim F_i, \widetilde{v}_i \sim \mathcal{B}_i(v'_i)}{\widetilde{v}_i =
	v^t_i} = f_i(t) \enspace.
\end{equation}
	
With this extra constraint, the values $\widetilde{\bm{w}}_i$ after the 
manipulation in the induced assignment problem of agent $i$ becomes
\begin{eqnarray*}
	\widetilde{w}^{st}_i & = & \e{\bm{v}_{-i} \sim F_{-i}, 
	\widetilde{\bm{v}}\sim\mathcal{B}(v^t_i, \bm{v}_{-i}), \bm{S} \sim 
	\mathcal{A}(\widetilde{\bm{v}})}{v^s_i(S_i)} \\
	& = & \e{\widetilde{\bm{v}}_{-i} \sim F_{-i}, \widetilde{v}_i \sim 
	\mathcal{B}_i(v^t_i), \bm{S} \sim \mathcal{A}(\widetilde{\bm{v}})}{v^s_i(S_i)} \\
	& = & \e{\bm{v}_{-i} \sim F_{-i}, \widetilde{v}_i \sim 
	\mathcal{B}_i(v^t_i), \bm{S} \sim \mathcal{A}(\widetilde{v}_i, \bm{v}_{-i})}
	{v^s_i(S_i)} \enspace.
\end{eqnarray*}

Thus, from the Bayesian viewpoint of agent $i$, the intermediate algorithms
$\mathcal{B}_{-i}$ of other agents are transparent. This property enables us to 
manipulate the structure of each assignment problem separately.

\subsection{Black-box reduction.}

Given an algorithm $\mathcal{A}$, the black-box reduction for social welfare will 
convert algorithm $\mathcal{A}$ into the following mechanism $\mathcal{M}_\mathcal{A}$:

\begin{enumerate}
	\item For each agent $i$, $1 \leq i \leq n$ 
		\hspace*{\fill} {\bf (Pre-computation)}
		\begin{enumerate}
			\item Estimate the values
				$\bm{w}_i = \{w^{st}_i\}_{1 \leq s, t \leq \ell}$ 
				of the induced assignment problem of $i$ with respect to 
				algorithm $\mathcal{A}$. Let 
				$\hat{\bm{w}}_i = \{\hat{w}^{st}_i\}_{1 \leq s, t\leq \ell}$
				denote the estimated values.
			\item Find the social welfare maximizing allocation 
				$\bm{x}_i = \{x^{st}_i\}_{1 \leq s, t \leq \ell}$
				and the corresponding envy-free prices 
				$\bm{p}_i = (p^1_i, \dots, p^\ell_i)$ for the induced 
				assignment problem of agent $i$ with estimated values.
		\end{enumerate}
	\item Manipulate the valuations with intermediate algorithms
		$\mathcal{B}_i$, $1 \leq i \leq n$, as follows:
		\hspace*{\fill} {\bf (Decoupling)}
		\begin{enumerate}
			\item[] Suppose the reported valuation of agent $i$ is
				$v'_i = v^s_i$, $1 \leq i \leq n$.
				Let $\widetilde{v}_i = \mathcal{B}_i(v'_i) = v^t_i$ 
				with probability $x^{st}_i / f_i(s)$ for $1 \leq t \leq \ell.$
		\end{enumerate}
	\item Allocate services according to $\mathcal{A}(\widetilde{\bm{v}})$. 
		\hspace*{\fill} {\bf (Allocation)}
		\begin{enumerate}
			\item Let $\bm{S} = (S_1, \dots, S_n)$ denote the allocation by
				algorithm $\mathcal{A}$ with input $\widetilde{\bm{v}}$.
			\item For each agent $i$, suppose the reported valuation is 
				$v'_i = v^s_i$ and the manipulated valuation is 
				$\widetilde{v}_i = v^t_i$, charge agent $i$ with price 
				$p^t_i v^s_i(S_i) / \hat{w}^{st}_i$.
		\end{enumerate}
\end{enumerate}

The following theorem
states that this reduction produces BIC while preserving the performance
with respect to social welfare.

\begin{theorem}
	\label{thm:welfarereduction}
	Suppose $\mathcal{A}$ is an algorithm for a multi-parameter mechanism design
	problem $\langle \bm{I}, \bm{J}, \bm{V}, \bm{F} \rangle$.
	\begin{enumerate}
		\item If the estimated values $\hat{\bm{w}}_i$ are 
			accurate, then mechanism $\mathcal{M}_\mathcal{A}$ is BIC, IR, and 
			guarantees at least $SW^\mathcal{A}$ of social welfare. 
		\item If the estimated values $\hat{\bm{w}}_i$ satisfy that for any 
			$1 \leq s, t \leq \ell$, 
			$\hat{w}^{st}_i \in [(1 - \epsilon) w^{st}_i, 
			(1 + \epsilon) w^{st}_i]$, 
			then mechanism $\mathcal{M}_\mathcal{A}$ is $4 \epsilon v_{max}$-BIC, IR, 
			and guarantees at least $(1 - 2\epsilon) \cdot SW^\mathcal{A}$ of 
			social welfare.
		\item If the estimated values $\hat{\bm{w}}_i$ 
			satisfy that for any $1 \leq s, t \leq \ell$,
			$\hat{w}^{st}_i \in [w^{st}_i - \epsilon, 
			w^{st}_i + \epsilon]$, then mechanism $\mathcal{M}_\mathcal{A}$ is 
			$4 \epsilon$-BIC, IR, and guarantees at least 
			$SW^\mathcal{A} - 2 n \epsilon$ of social welfare.
	\end{enumerate}
\end{theorem}

Let us illustrate the proof of part 1. The proofs of the other two parts are 
tedious and simple calculations along the same line. We will omit these proofs
in this extended abstract.

\begin{proof}
	We consider the case when the estimated values $\hat{\bm{w}}_i$ 
	are accurate,
	that is, $\hat{w}^{st}_i = w^{st}_i$ for all $1 \leq i \leq n$ and $1 \leq s, t \leq \ell$.

	\paragraph{Individual rationality.}
	By our choice of envy-free prices, we have that 
	$p^t_i \leq w^{st}_i$ for all $1 \leq i \leq n$ and $1 \leq s, t \leq \ell$.
	Thus, we always guarantee
	$$p^t_i \, \frac{v^s_i(S_i)}{w^{st}_i} \leq v^s_i(S_i) \enspace.$$ 

	So the mechanism $\mathcal{M}_\mathcal{A}$ that we get from the reduction
	always provides non-negative utilities for the agents.
	Essentially the same proof also shows IR for part 2 and 3.	
	
	\paragraph{Bayesian incentive compatibility.}
	We will first show that the intermediate algorithms in the decoupling
	step of the reduction satisfy constraint \eqref{eq:interconstraint}.
	Let $\bm{x}_i$ denote the social welfare maximizing allocation
	that the reduction finds
	for the induced assignment problem of agent $i$ for $1 \leq i \leq n$.
	Note that these social welfare maximizing allocations are market-clearing. 
	We have that if the reported valuation $v'_i$ follows the distribution 
	$F_i$, then the distribution of the manipulated valuation 
	$\widetilde{v}_i$ satisfies that
	$$\pr{}{\widetilde{v}_i = v^t_i} = \sum_{s=1}^\ell \pr{}{v'_i = v^s_i} \,
	\pr{}{\widetilde{v}_i = v^t_i : v'_i = v^s_i} 
	= \sum_{s=1}^\ell f_i(s) \, \frac{x^{st}_i}{f_i(s)} 
	= \sum_{s=1}^\ell x^{st}_i = f_i(t) \enspace.$$

	Indeed, the intermediate algorithms satisfy constraint
	\eqref{eq:interconstraint}. Thus, for each $1 \leq i \leq n$ 
	the intermediate algorithm $\mathcal{B}_i$
	only changes the structure of induced assignment problem of agent $i$
	and leaves the induced assignment problems of other agents untouched.
	
	Next, we will verify that in each of the manipulated assignment 
	problem, the identity allocation maximizes the social welfare and the
	prices are the corresponding envy-free prices.

	For each agent $i$,
	we let $\widetilde{\bm{w}}_i = 
	\{\widetilde{w}^{st}_i\}_{1 \leq s, t \leq \ell}$
	and $\widetilde{\bm{p}}_i = 
	(\widetilde{p}^1_i, \dots, \widetilde{p}^\ell_i)$
	denote the values and the expected prices of the virtual products 
	respectively in
	the manipulated assignment problem of agent $i$. We have that for any 
	$1 \leq r, s \leq \ell$, 
	\begin{align*}
		\widetilde{w}^{rs}_i & = 
		\sum_{t=1}^\ell \pr{}{\widetilde{v}_i = v^t_i} \e{\bm{v}_{-i}, 
		\bm{S} \sim \mathcal{A}(v^t_i, \bm{v}_{-i})}{v^r_i(S_i)} 
		= \sum_{t=1}^\ell \frac{x^{st}_i}{f_i(s)} \, w^{rt}_i \enspace, \\
		\widetilde{p}^s_i 
		& = \sum_{t=1}^\ell \pr{}{\widetilde{v}_i = v^t_i} 
		\e{\bm{v}_{-i}, \bm{S} \sim \mathcal{A}(v^t_i, \bm{v}_{-i})}
		{p^t_i \frac{v^s_i(S_i)}{w^{rs}_i}} 
		= \sum_{t=1}^\ell \frac{x^{st}_i}{f_i(s)} \, p^t_i\enspace.
	\end{align*}

	Thus, in the manipulated assignment problem of agent $i$,
	the utility of the virtual buyer $r$ of the virtual product $s$,
	$1 \leq r, s \leq \ell$, is 
	$$\widetilde{w}^{rs}_i - \widetilde{p}^s_i = \sum_{t=1}^\ell
	\frac{x^{st}_i}{f_i(s)} \, (w^{rt}_i - p^t_i) 
	\leq \sum_{t=1}^\ell \frac{x^{st}_i}{f_i(s)} \, \max_k 
	\{w^{rk}_i - p^k_i\} = \max_k \{w^{rk}_i - p^k_i\} \enspace.$$

	Since $\bm{p}_i$ are chosen to be the envy-free prices, we have that
	$x^{rt}_i > 0$ only if $w^{rt}_i - p^t_i = \max_k \{w^{rk}_i - p^k_i\}$.
	Hence, when agent $i$ reports its valuation truthfully, that is, $r = s$,
	the above inequality holds with equality. 
	So the $\widetilde{p}_i$ are envy-free
	prices with respect to the identity allocation $\widetilde{\bm{x}}_i$
	of the manipulated assignment problem of agent $i$. By Lemma 
	\ref{lemma:assignment1}
	we know the allocation $\widetilde{\bm{x}}_i$ maximizes the social welfare.
	Thus, mechanism $\mathcal{M}_\mathcal{A}$ is BIC according to Lemma 
	\ref{lemma:characterization}.
	
	\paragraph{Social welfare.} The expected social welfare for this 
	mechanism is $\sum_{i=1}^n \sum_{s=1}^\ell \sum_{t=1}^\ell
	x^{st}_i w^{st}_i$. Since for any $1 \leq i \leq n$
	the allocation $\bm{x}_i$ maximizes the social welfare for the 
	induced assignment problem of agent $i$, the social welfare of $\bm{x}_i$
	is at least as large as that of the identity allocation, that is,
	$$\forall i : \sum_{s=1}^\ell \sum_{t=1}^\ell x^{st}_i w^{st}_i \geq 
	\sum_{s=1}^\ell f_i(s) w^{ss}_i 
	= \e{\bm{v} \sim \bm{F}, \bm{S} \sim \mathcal{A}(\bm{v})}{v_i(S_i)} 
	\enspace.$$

	Thus, we have that
	$$SW^{\mathcal{M}_\mathcal{A}} =
	\sum_{i=1}^n \sum_{s, t=1}^\ell x^{st}_i w^{st}_i
	\geq \sum_{i=1}^n \e{\bm{v} \sim \bm{F}, \bm{S} \sim \mathcal{A}(\bm{v})}
	{v_i(S_i)} = \e{\bm{v} \sim \bm{F}, \bm{S} \sim \mathcal{A}(\bm{v})}
	{\sum_{i=1}^n v_i(S_i)} = SW^\mathcal{A} \enspace.$$
\end{proof}

\subsection{Estimating values by sampling.}

There is still one technical issue that we need to settle in the reduction.
In this section, we will briefly discuss how to use the standard sampling
technique to obtain good estimated values of 
$\bm{w}_i = \{w^{st}_i\}_{1 \leq s, t \leq \ell}$ 
for the induced assignment problem of agent $i$ for $1 \leq i \leq n$. 

By definition, $w^{st}_i$ is the expectation of a random
variable $v^s_i(S_i)$, where $S_i$ is the allocated service given by 
$\mathcal{A}$ over random realization of the valuations $\bm{v}_{-i}$ of other agents
and random coin flips of the algorithm. Hence, if the standard deviation 
of $v^s_i(S_i)$ is not too large compared to its mean (no more than a 
polynomial factor), then we can draw polynomially many independent samples and 
take the average value as our estimated value. 
More concretely, the sampling algorithm proceeds as follows.
\begin{enumerate}
	\item Draw $N = 4 \, c^2 \log(n \ell^2 /\epsilon) / \epsilon^2$ independent 
		samples of $\bm{v} \sim \bm{F}$ conditioned on that the valuation
		of agent $i$ is $v^t_i$, where 
		$$c = \frac{\sd{\bm{v}_{-i}, \bm{S} \sim \mathcal{A}(v^t_i, \bm{v}_{-i})}
		{v^s_i(S_i)}}
		{\e{\bm{v}_{-i}, \bm{S} \sim \mathcal{A}(v^t_i, \bm{v}_{-i})}{v^s_i(S_i)}}
		\enspace.$$
		Let $\bm{v}^1, \dots, \bm{v}^N$ denote these $N$ sample.
	\item Use algorithm $\mathcal{A}$ to compute an allocation $\bm{S}^k \sim 
		\mathcal{A}(\bm{v}^k)$ for each sample $\bm{v}^k$, $1 \leq k \leq N$.
	\item Let $\hat{w}^{st}_i$ be the average of $v^s_i(S^k_i)$, $1 \leq k
		\leq N$.
\end{enumerate}

\begin{lemma}
	\label{lemma:relative}
	The estimated values $\hat{\bm{w}}_i$, $1 \leq i \leq n$, by the above
	sampling procedure satisfy for any $1 \leq i \leq n$ and $1 \leq s, t \leq 
	\ell$,
	$$\hat{w}^{st}_i \in \left[(1 - \epsilon) w^{st}_i,
	(1 + \epsilon) w^{st}_i\right]$$
	with probability at least $1 - \epsilon$.
\end{lemma}

\begin{proof}
	We shall have that 
	\begin{align*}
		\e{}{\hat{w}^{st}_i} & = \e{\bm{v}_{-i}, \bm{S} \sim 
		\mathcal{A}(v^t_i,\bm{v}_{-i})}{v^s_i(S_i))} = w^{st}_i \enspace, \\
		\sd{}{\hat{w}^{st}_i} & = \frac{1}{\sqrt{N}} \, \sd{\bm{v}_{-i}, \bm{S} 
		\sim \mathcal{A}(v^t_i,\bm{v}_{-i})}{v^s_i(S_i)} 
		= \frac{c}{\sqrt{N}} \, \e{}{\hat{w}^{st}_i}
		= \frac{c}{\sqrt{N}} \, w^{st}_i \enspace.
	\end{align*}
	
	By Chernoff-Hoeffding bound, we get
	\begin{eqnarray*}
		\pr{}{\left|\hat{w}^{st}_i - w^{st}_i\right| > 
		\epsilon w^{st}_i} 
		& = & \pr{}{\left|\hat{w}^{st}_i - \e{}{\hat{w}^{st}_i}\right| > 
		\frac{\epsilon \sqrt{N}}{c} \, \sd{}{\hat{w}^{st}_i}} \\
		& = & \pr{}{\left|\hat{w}^{st}_i - \e{}{\hat{w}^{st}_i}\right| > 
		2 \, \sqrt{\log{(n \ell^2/\epsilon)}} \, \sd{}{\hat{w}^{st}_i}} \\
		& \leq & e^{- \log{(n \ell^2 /\epsilon)}} = \frac{\epsilon}{n \ell^2}
		\enspace.
	\end{eqnarray*}

	Since we only need to estimate $n \ell^2$ values,
	by union bound we get that with probability at least $1 - \epsilon$
	the estimated value $\hat{w}^{st}_i$ is within
	$\epsilon$ relative error compared to $w^{st}_i$ for all
	$1 \leq i \leq n$ and $1 \leq s, t \leq \ell$.
\end{proof}

Thus, if the allocation algorithm $\mathcal{A}$ admits $SW^\mathcal{A}$ social welfare and
the ratio $c$ is only polynomially large, then by part 2 of Theorem 
\ref{thm:welfarereduction} we get that mechanism $\mathcal{M}_\mathcal{A}$ gives
$(1-3\epsilon) \cdot SW^\mathcal{A}$ social welfare and is $4\epsilon v_{max}$-BIC.
The running time is polynomial in the input size and $1 / \epsilon$, assuming
a black-box call to algorithm $\mathcal{A}$ can be done in a single step. 
In other words, we get a FPTAS reduction.

\medskip

The next lemma gives an alternative bound of the sampling error with respect
to absolute error.

\begin{lemma}
	\label{lemma:absolute}
	If we draw $N' = 4\log(n\ell^2/\epsilon)/\epsilon^2$ independent samples, 
	then with probability at least $1-\epsilon$ the estimated values 
	$\hat{w}^{st}_i \in [w^{st}_i - \epsilon v_{max}, 
	w^{st}_i + \epsilon v_{max}]$ for all $1 \leq i \leq n$ and $1 \leq s, t
	\leq \ell$.
\end{lemma}

\begin{proof}
	In this case, we have 
	\begin{align*}
		\e{}{\hat{w}^{st}_i} & = \e{\bm{v}_{-i}, \bm{S} \sim 
		\mathcal{A}(v^t_i,\bm{v}_{-i})}{v^s_i(S_i))} = w^{st}_i \enspace, \\
		\sd{}{\hat{w}^{st}_i} & = \frac{1}{\sqrt{N'}} \, \sd{\bm{v}_{-i}, 
		\bm{S}\sim\mathcal{A}(v^t_i,\bm{v}_{-i})}{v^s_i(S_i)} 
		\leq \frac{1}{\sqrt{N'}}\max_{S_i} \, v^s_i(S_i) \leq \frac{1}
		{\sqrt{N'}} \, v_{max}
		\enspace.
	\end{align*}
	
	By Chernoff bound we get that
	\begin{eqnarray*}
		\pr{}{\left|\hat{w}^{st}_i - w^{st}_i\right| > 
		\epsilon v_{max}}
		& \leq & \pr{}{\left|\hat{w}^{st}_i - \e{}{\hat{w}^{st}_i}\right| > 
		\frac{\epsilon}{\sqrt{N'}} \sd{}{\hat{w}^{st}_i}} \\
		& = & \pr{}{\left|\hat{w}^{st}_i - \e{}{\hat{w}^{st}_i}\right| > 
		2 \, \sqrt{\log{(n \ell^2/\epsilon)}} \, \sd{}{\hat{w}^{st}_i}} \\
		& \leq & e^{- \log{(n \ell^2 /\epsilon)}} = \frac{\epsilon}{n \ell^2}
		\enspace.
	\end{eqnarray*}

	By union bound, we have $\hat{w}^{st}_i \in [w^{st}_i - \epsilon v_{max},
	w^{st}_i + \epsilon v_{max}]$ for all $1 \leq i \leq n$ and $1 \leq s, t
	\leq \ell$.
\end{proof}

Suppose the ratio $v_{max} / SW^\mathcal{A}$ is upper bounded by a polynomial of 
the input size.
Then, if we choose $\epsilon = \delta \, SW^\mathcal{A} / 2n v_{max}$ in the above
lemma, we will get that
$$\left|\hat{w}^{st}_i - w^{st}_i\right| < \delta \, SW^\mathcal{A} / 2n \enspace.$$

By part 3 of Theorem {\ref{thm:welfarereduction} we obtain that
mechanism $\mathcal{M}_\mathcal{A}$ provides at least $(1 - \delta) SW^\mathcal{A}$ of social
welfare and is $4 \epsilon$-BIC and IR.
The running time is polynomial in the input size and $1 / \delta$.

\section{Reductions for revenue and residual surplus}

In the reduction for social welfare in the previous section, 
we only consider market-clearing allocations in the induced assignment 
problems. This is because for any agent 
$i$, we want to make sure that the intermediate algorithm $\mathcal{B}_i$ 
is transparent to all agents except agent $i$. If we restrict ourselves to 
market-clearing allocations, we do not know any way to get reasonable bounds 
on revenue and residual surplus.

However, if we focus on an important sub-class of multi-parameter mechanism
design problems that includes the combinatorial auction problem and its special
cases, then we have some flexibility in choosing the allocations for the 
induced assignment problem and obtain theoretical bounds on revenue and
residual surplus. More concretely, we will consider mechanism design problems
that are {\em downward-closed}. We let $\phi$ denote the null service so that
allocating $\phi$ to an agent implies that agent is not served, that is,
$v_i(\phi) = 0$ for all $1 \leq i \leq n$. 

\begin{definition}
	A multi-parameter mechanism design problem 
	$\langle \bm{I}, \bm{J}, \bm{V}, \bm{F} \rangle$
	is {\em downward-closed} if for any 
	feasible allocation $\bm{S} = (S_1, \dots, S_n) \in \bm{J}$ and any
	$1 \leq i \leq n$, the allocation 
	$(S_1, \dots, S_{i-1}, \phi, S_{i+1}, \dots, S_n)$ is also feasible.
\end{definition}

We let $\delta = \min \{f_i(s) : 1 \leq i \leq n, 1 \leq s \leq \ell, 
f_i(s) > 0\}$ denote the granularity of the prior distributions. We will 
prove the following result.

\begin{theorem}
	\label{thm:otherreduction}
	For any algorithm $\mathcal{A}$, there is a mechanism that is IR, BIC, and 
	provides at least $\Omega(SW^\mathcal{A} / \log(1 / \delta))$ of revenue 
	(residual surplus).
\end{theorem}

\subsection{Meta-reduction.}

We will first introduce a meta-reduction scheme based on algorithms that
compute envy-free solutions for fractional assignment problems.
Suppose $\mathcal{C}$ is an algorithm that computes envy-free solutions
$(\bm{x}, \bm{p})$ for any given fractional assignment problem.
Let $\mathcal{A}$ be an algorithm for a multi-parameter mechanism design problem
$\langle \bm{I}, \bm{J}, \bm{V}, \bm{F} \rangle$.
We will convert algorithm $\mathcal{A}$ into to a mechanism $\mathcal{M}^\mathcal{C}_\mathcal{A}$:

\begin{enumerate}
	\item For each agent $i$ \hspace*{\fill} {\bf (Pre-computation)}
		\begin{enumerate}
			\item Estimate the values 
				$\bm{w}_i = \{w^{st}_i\}_{1 \leq s, t \leq \ell}$ 
				for the induced assignment problem of agent $i$
				with respect to $\mathcal{A}$. Let $\hat{\bm{w}}_i = 
				\{\hat{w}^{st}_i\}_{1 \leq s, t \leq \ell}$ denote the
				estimated values.
			\item Use $\mathcal{C}$ to solve the induced assignment problems 
				with estimated values. Let
				$(\bm{x}_i, \bm{p}_i)$ denote the solution by $\mathcal{C}$ for
				the induce assignment problem of agent $i$.
			\item Let $y^t_i = f_i(t) - \sum_{s=1}^\ell x^{st}_i$ denote the 
				unallocated supply of virtual product $t$ in solution 
				$(\bm{x}_i, \bm{p}_i)$ for all $1 \leq i \leq n$ and 
				$1 \leq t \leq \ell$. 
			\item Let $y_i = \sum_{t=1}^\ell y^t_i$ 
				denote the total amount of unallocated virtual products 
				in $(\bm{x}_i, \bm{p}_i)$ for all $1 \leq i \leq n$.
		\end{enumerate}
	\item Manipulate the valuations with intermediate algorithm $\mathcal{B}_i$,
		$1 \leq i \leq n$, as follows:
		\hspace*{\fill} {\bf (Decoupling)}
		\begin{enumerate}
			\item Suppose the reported valuation of agent $i$ is $v'_i = 
				v^s_i$.
			\item Let $\widetilde{v}_i = \mathcal{B}_i(v'_i) = v^t_i$ with 
				probability $x^i_{st} / f_i(s)$ for $1 \leq t \leq \ell$.
			\item With probability $1 - \sum_t x^{st}_i / f_i(s)$, the
				manipulated valuation $\widetilde{v}_i$ is unspecified 
				in the previous step. In this
				case, let $\widetilde{v}_i = v^t_i$ with probability 
				$y^t_i / y_i$ for $1 \leq t \leq \ell$.
		\end{enumerate}
	\item Allocate services as follows: \hspace*{\fill} {\bf (Allocation)}
		\begin{enumerate}
			\item Compute a tentative allocation
				$\widetilde{\bm{S}} = 
				(\widetilde{S}_1, \dots, \widetilde{S}_n) = 
				\mathcal{A}(\widetilde{\bm{v}})$.
			\item For each agent $i$, let $S_i = \widetilde{S}_i$ if the
				manipulated valuation $\widetilde{v}_i$ is specified in
				step 2b). Let $S_i = \phi$ otherwise. Allocate services
				according to $\bm{S}$.
			\item For each agent $i$, suppose the reported valuation is 
				$v'_i = v^s_i$ and the manipulated valuation is 
				$\widetilde{v}_i = v^t_i$, charge agent $i$ with price 
				$p^t_i v^s_i(S_i) / \hat{w}^{st}_i$.
		\end{enumerate}
\end{enumerate}

The following theorem states the above meta-reduction scheme converts 
algorithms into IR and BIC mechanisms.

\begin{theorem}
	Suppose the algorithm $\mathcal{C}$ always provides envy-free solutions.
	\begin{enumerate}
		\item If the estimated values $\hat{\bm{w}}_i$
			are accurate, then mechanism $\mathcal{M}^\mathcal{C}_\mathcal{A}$ is IR and BIC.
		\item If the estimated values $\hat{\bm{w}}_i$ satisfy that for
			any $1 \leq s, t \leq \ell$, $\hat{w}^{st}_i \in [(1 - \epsilon) 
			w^{st}_i, (1 + \epsilon) w^{st}_i]$, then $\mathcal{M}^\mathcal{C}_\mathcal{A}$ is 
			IR and $4 \epsilon v_{max}$-BIC.
		\item If the estimated values $\hat{\bm{w}}_i$ satisfy that
			for any $1 \leq s, t \leq \ell$, $\hat{w}^{st}_i \in 
			[w^{st}_i - \epsilon, w^{st}_i + \epsilon]$, then 
			$\mathcal{M}^\mathcal{C}_\mathcal{A}$ is IR and $4 \epsilon$-BIC.
	\end{enumerate}
\end{theorem}

\begin{proof}
	Let us outline the proof for part 1. Proofs of the other two parts are
	calculations along the same line.

	Note that $p^t_i \leq w^{st}_i$ for all $1 \leq i \leq n$ and 
	$1 \leq s, t \leq \ell$.
	The mechanism is IR because for any $1 \leq i \leq n$ and 
	$1 \leq s \leq \ell$ the utility for an agent $i$ with valuation $v^s_i$ 
	in any realization is
	$$v^s_i(S_i) - p^t_i \frac{v^s_i(S_i)}{w^{st}_i} \geq 0 \enspace.$$
	
	Next, we will show that mechanism $\mathcal{M}^\mathcal{C}_\mathcal{A}$ is BIC.
	We first verify that the intermediate algorithms $\mathcal{B}_i$,
	$1 \leq i \leq n$, satisfy the constraint \eqref{eq:interconstraint}. 
	For any agent $i$, if its valuation $v_i$
	is drawn from distribution $F_i$, then the probability that the
	manipulated valuation $\widetilde{v}_i = \mathcal{B}_i(v_i) = v^t_i$ is 
	\begin{eqnarray*}
		\sum_{s=1}^\ell f_i(s) \, \left[\frac{x^{st}_i}{f_i(s)} + 
		\left(1 - 
		\sum_{r=1}^\ell \frac{x^{sr}_i}{f_i(s)} \right) \, \frac{y^t_i}{y_i}
		\right]
		& = & \sum_{s=1}^\ell x^{st}_i + \left( \sum_{s=1}^\ell f_i(s) - 
		\sum_{s=1}^\ell \sum_{r=1}^\ell x^{sr}_i \right) \, \frac{y^t_i}{y_i} \\
		& = & \sum_{s=1}^\ell x^{st}_i + \left( \sum_{r=1}^\ell f_i(r) - 
		\sum_{r=1}^\ell \sum_{s=1}^\ell x^{sr}_i \right) \, \frac{y^t_i}{y_i} \\
		& = & \sum_{s=1}^\ell x^{st}_i + \sum_{r=1}^\ell \left( f_i(r) - 
		\sum_{s=1}^\ell x^i_{sr} \right) \, \frac{y^t_i}{y_i} \\
		& = & \sum_{s=1}^\ell x^{st}_i + \sum_{r=1}^\ell y^r_i \,
		\frac{y^t_i}{y_i} = \sum_{s=1}^\ell x^{st}_i + y^t_i 
		= f_i(t) \enspace.
	\end{eqnarray*}

 	Thus, we get that for each agent $i$, the intermediate algorithms 
	$\mathcal{B}_j$, $1 \leq j \leq n$ and $j \neq i$, are transparent to it.
	So the expected value of agent $i$ of the service it gets, when
	its genuine valuation is $v_i = v^s_i$ and the manipulate valuation, is 
	$\widetilde{v}_i = v^t_i$ is exactly 
	$$w^{st}_i = \e{\bm{v}_{-i}, \bm{S} \sim \mathcal{A}(v^t_i, \bm{v}_{-i})}
	{v^s_i(S_i)} \enspace.$$ 

	Hence, the expected value of agent $i$ of the servie it gets, when its
	genuine valuation is $v_i = v^s_i$ and the reported valuation is $v'_i
	= v^t_i$, is 
	$$\widetilde{w}^{st}_i = 
	\sum_{r=1}^\ell \frac{x^{tr}_i}{f_i(t)} \, w^{sr}_i \enspace.$$
	
	And the expected price for agent $i$ when the reported valuation is 
	$v'_i = v^t_i$ is 
	$$\widetilde{p}^t_i = 
	\sum_{r=1}^\ell \frac{x^{tr}_i}{f_i(t)} \, \e{\bm{v}_{-i}, \bm{S} \sim 
	\mathcal{A}(v^r_i, \bm{v}_{-i})}{p^r_i \, \frac{v^t_i(S_i)}{w^{tr}_i}} 
	= \sum_{r=1}^\ell \frac{x^{tr}_i}{f_i(t)} p^r_i \enspace.$$

	Thus, the the expected utility of agent $i$, when its genuine valuation
	is $v_i = v^s_i$ and its reported valuation is $v'_i = v^t_i$, is 
	$$\widetilde{w}^{st}_i - \widetilde{p}^t_i = \sum_{r=1}^\ell
	\frac{x^{tr}_i}{f_i(t)} \, (w^{sr}_i - p^r_i) 
	\leq \sum_{r=1}^\ell \frac{x^{tr}_i}{f_i(t)} \, \max_k 
	\{w^{sk}_i - p^k_i\} = \max_k \{w^{sk}_i - p^k_i\} \enspace.$$

	Since $\bm{p}_i$ are chosen to be the envy-free prices, we have that
	$x^{sr}_i > 0$ only if $w^{sr}_i - p^r_i = \max_k \{w^{sk}_i - p^k_i\}$.
	Hence, when agent $i$ reports its valuation truthfully, that is, $s = t$,
	the above inequality if tight. Moreover, the above utility is always
	non-negative. So mechanism $\mathcal{M}^\mathcal{C}_\mathcal{A}$ is BIC.
\end{proof}

Moreover, the revenue and residual surplus of mechanism 
$\mathcal{M}^\mathcal{C}_\mathcal{A}$ 
is related to the social welfare and revenue of the induced assignment 
problems as stated in following proposition.

\begin{proposition}
	The expected revenue (residual surplus) of the mechanism 
	$\mathcal{M}^\mathcal{C}_\mathcal{A}$ equals the sum of the revenue (residual surplus) 
	of the manipulated assignment problems.
\end{proposition}

By choosing proper allocation algorithm $\mathcal{C}$, we can obtain
theoretical bounds for the revenue or residual surplus in the manipulated
induced assignment problems and thus theoretical bounds for mechanism
$\mathcal{M}^\mathcal{C}_\mathcal{A}$.

\subsection{Assignment algorithms.}

In this section, we will introduce two algorithms for computing envy-free 
solutions for the induced assignment problems. These two algorithms
provides theoretical bounds for revenue and residual surplus.

\paragraph{Revenue.} 
The first algorithm provides revenue that is a $\Omega(1/\log(1/\delta))$
fraction of $SW^\mathcal{A}$, the social welfare by algorithm $\mathcal{A}$. The idea is
to introduce proper reserve prices to the induced assignment problems by
redundant virtual buyers. This is inspired by the techniques by
Guruswami et al.~\cite{guruswami2005profit}.
For the induced assignment problem of agent $i$, $1 \leq i \leq n$, 
the assignment algorithm $\mathcal{C}_R$ for revenue maximization proceeds as
follows:
\begin{enumerate}
	\item Find the social welfare maximizing allocation $\bm{x}_i = 
		\{x^{st}_i\}_{1 \leq s, t \leq \ell}$.
	\item Suppose $u_{max}$ is the maximal valuation 
		among the virtual buyer-product pair $(s, t)$ with non-zero $x^{st}_i$,
		that is, 
		$$u_{max} = \max \{w^{st}_i : 1 \leq s, t \leq \ell, x^{st}_i > 0\}
		\enspace.$$
	\item Recall that $\delta = \min \{f_i(t) : 1 \leq i \leq n, 1 \leq t \leq
		\ell, f_i(t) > 0\}$ denotes the granularity of the prior distribution. 
		For $1 \leq k \leq \log(1 / \delta)$:
		\begin{enumerate}
			\item Consider the following variant
				of the induced assignment instance of agent $i$: \\
				For each virtual product $1 \leq t \leq \ell$, add a 
				dummy virtual buyer with demand $1 + \delta$ and value 
				$u_k = u_{max} / 2^k$ for virtual product $t$
				and value $0$ for other virtual products.
			\item Find social welfare maximizing allocation $\bm{x}_{ik}$ and 
				envy-free prices $\bm{p}_{ik}$ for this variant.
			\item Let $(\hat{\bm{x}}_{ik}, \hat{\bm{p}}_{ik})$ be the projection
				of $(\bm{x}_{ik}, \bm{p}_{ik})$ on the original induced 
				assignment problem of agent $i$, that is, for any 
				$1 \leq s, t \leq \ell$,
				$\hat{x}^{st}_{ik} = x^{st}_{ik}$, $\hat{p}^t_{ik} = p^t_{ik}$.
		\end{enumerate}
	\item Return the $(\hat{\bm{x}}_{ik}, \hat{\bm{p}}_{ik})$, 
		$1 \leq k \leq \log(1/\delta)$, with best revenue.
\end{enumerate}

\begin{lemma}
	Assignment algorithm $\mathcal{C}_R$ always return an envy-free solution 
	$(\bm{x}, \bm{p})$. The revenue is at least a
	$\Omega(1/\log(1/\delta))$ fraction of the optimal social welfare of
	the assignment problem.
\end{lemma}

\begin{proof}
	The envy-freeness follows from the fact that 
	$(\hat{\bm{x}}_{ik}, \hat{\bm{p}}_{ik})$, $1 \leq k \leq \log(1/\delta)$,
	are projections of envy-free solutions and thus are also envy-free.
	
	Now we consider the revenue by $\mathcal{C}_R$. We let $r_k$ denote
	the revenue by solution $(\hat{\bm{x}}_{ik}, \hat{\bm{p}}_{ik})$.
	Note that in $(\hat{\bm{x}}_{ik}, \hat{\bm{p}}_{ik})$, all prices are at 
	least $u_k$ and the amount of virtual products that are sold is at least
	$\sum_{s, t : w^{st}_i \geq u_k} x^{st}_i$. Hence, we have 
	$$r_k \geq w_k \sum_{s, t : w_{st} \geq u_k} x^{st}_i \enspace.$$

	Note that if we extend the definition of $u_k$ and let $u_k = u_{max} / 
	2^k$ for all non-negative integer $k$, then we have
	\begin{eqnarray}
		\sum_{k=1}^\infty u_k \sum_{s, t : w^{st}_i \geq u_k} x^{st}_i 
		\notag 
		& = & \sum_{k=1}^\infty (u_{k - 1} - u_k) 
		\sum_{s, t : w^{st}_i \geq u_k} x^{st}_i \notag \\
		& = & \sum_{s=1}^\ell \sum_{t=1}^\ell x^{st}_i 
		\sum_{k : w^{st}_i \geq u_k} (u_{k-1} - u_k) \notag \\
		& = & \sum_{s=1}^\ell \sum_{t=1}^\ell x^{st}_i 
		\max_k \{u_{k-1} : w^{st}_i \geq u_k\} \notag \\
		& \geq & \sum_{s, t} x^{st}_i w^{st}_i \enspace. \label{eq:rev1}
	\end{eqnarray}

	On the other hand, the contribution of the tail is small compared to
	the social welfare.
	\begin{equation}
		\label{eq:rev2}
		\sum_{k=\log(1/\delta)+1}^{\infty} u_k 
		\sum_{s, t : w^{st}_i \geq u_k} x^{st}_i 
		\leq \sum_{k = \log(1/\delta) + 1}^{\infty} w_k 
		\leq \frac{\delta w_{max}}{2}  
		\leq \frac{\sum_{s, t} x^{st}_i w^{st}_i}{2} \enspace. 
	\end{equation}

	The last inequality holds because allocating the most valuable 
	virtual product the one of the virtual buyer is a feasible allocation.
	Hence, consider the difference of the above formulas, 
	$\eqref{eq:rev1} - \eqref{eq:rev2}$, and we get that
	$$\sum_{k = 1}^{\log(1/\delta)} r_k \geq \sum_{k = 1}^{\log(1/\delta)} u_k
	\sum_{s, t : w_{st} \geq u_k} x^{st}_i \geq 
	\frac{\sum_{s, t} x^{st}_i w^{st}_i}{2} \enspace. $$

	Thus, by pigeon-hole-principle at least one of the assignment 
	$(\hat{\bm{x}}_{ik}, \hat{\bm{p}}_{ik})$ provides 
	revenue that is at least a $1/2\log(1/\delta)$ fraction of 
	the social welfare.
\end{proof}

The above lemma leads to the following results for revenue maximization.

\begin{proposition}
	Suppose the social welfare given by allocation algorithm $\mathcal{A}$ is
	$SW^\mathcal{A}$, the mechanism $\mathcal{M}^{\mathcal{C}_R}_\mathcal{A}$ guarantees at least
	$\Omega(SW^\mathcal{A} / \log(1/\delta))$ of revenue.
\end{proposition}

\paragraph{Complementary lower bound.}
The approximation ratio with respect to $SW^\mathcal{A}$ is tight due to the 
following example. Consider the auction problem with only one agent and one
item. Suppose with probability $1 / 2^k$ the agent has value $2^k$ for the item 
for $k = 1, 2, \dots, \log(1 / \delta)$. And with probability $\delta$, the
agent has value $0$ for the item. In this example, the granularity
of the prior distribution is $\delta$. 
The optimal social welfare is 
$\sum_{k=1}^{\log(1/\delta)} \frac{1}{2^k} \, 2^k = \log(1/\delta)$. 
But no BIC mechanism can achieve revenue better than $1$.

\paragraph{Residual surplus.} 
We turn to the residual surplus maximization problem. Note that revenue
and residual surplus are symmetric in the induced assignment problems. 
We will use the following assignment algorithm $\mathcal{C}_{RS}$
based on the same idea we use for the revenue maximization algorithm.

The residual surplus maximizing envy-free algorithm $\mathcal{C}_{RS}$ is as
follows:
\begin{enumerate}
	\item Find the social welfare maximizing allocation $\bm{x}_i = 
		\{x^{st}_i\}_{1 \leq s, t \leq \ell}$.
	\item Suppose $u_{max}$ is the maximal valuation 
		among the virtual buyer-product pair $(s, t)$ with non-zero $x^{st}_i$,
		that is, 
		$$u_{max} = \max \{w^{st}_i : 1 \leq s, t \leq \ell, x^{st}_i > 0\}
		\enspace.$$
	\item Recall that $\delta = \min \{f_i(t) : 1 \leq i \leq n, 1 \leq t \leq
		\ell, f_i(t) > 0\}$ denotes the granularity of the prior distribution. 
		For $1 \leq k \leq \log(1 / \delta)$:
		\begin{enumerate}
			\item Consider the following variant
				of the induced assignment instance of agent $i$: \\
				For each virtual buyer $1 \leq t \leq \ell$, add a 
				dummy virtual product with demand $1 + \delta$ and value 
				$u_k = u_{max} / 2^k$ for virtual buyer $t$
				and value $0$ for other virtual buyer.
			\item Find social welfare maximizing allocation $\bm{x}_{ik}$ and 
				envy-free prices $\bm{p}_{ik}$ for this variant.
			\item Let $(\hat{\bm{x}}_{ik}, \hat{\bm{p}}_{ik})$ be the projection
				of $(\bm{x}_{ik}, \bm{p}_{ik})$ on the original induced 
				assignment problem of agent $i$, that is, for any 
				$1 \leq s, t \leq \ell$,
				$$\hat{x}^{st}_{ik} = x^{st}_{ik} \quad , \quad 
				\hat{p}^t_{ik} = p^t_{ik} \enspace.$$
		\end{enumerate}
	\item Return the $(\hat{\bm{x}}_{ik}, \hat{\bm{p}}_{ik})$, 
		$1 \leq k \leq \log(1/\delta)$, with best revenue.
\end{enumerate}

The proofs of the following lemma and theorem is almost identical to the 
revenue maximization part so we omit the proofs here.

\begin{lemma}
	Assignment algorithm $\mathcal{C}_{RS}$ always return an envy-free solution 
	$(\bm{x}, \bm{p})$. The residual surplus is at least a
	$\Omega(1/\log(1/\delta))$ fraction of the optimal social welfare of
	the assignment problem.
\end{lemma}

\begin{proposition}
	Suppose the social welfare given by allocation algorithm $\mathcal{A}$ is
	$SW^\mathcal{A}$, the mechanism $\mathcal{M}^{\mathcal{C}_{RS}}_\mathcal{A}$ guarantees at least
	$\Omega(SW^\mathcal{A} / \log(1/\delta))$ of residual surplus.
\end{proposition}

\section{Application in combinatorial auctions}

In this section we will briefly illustrates how to use the reduction for social 
welfare in this paper to derive a combinatorial auction mechanism that matches 
the best algorithmic result. 

\paragraph{Combinatorial auctions.} In the combinatorial auctions, we consider
a principal with $m$ items (exactly one copy of each item) and $n$ agents. 
Each agent has a private valuation $v_i \sim F_i$ for subsets of items. 
The goal is to design a protocol to 
allocate the items and to charge prices to the agents.

\medskip

We will show the following corollaries of our reduction for social welfare.

\begin{corollary}
	For combinatorial auctions with sub-additive (or fractionally sub-additive) 
	agents where the prior distributions have finite and poly-size 
	supports, there is a $\left(\frac{1}{2} - \epsilon\right)$-approximate
	(or $\left(1- \frac{1}{e}-\epsilon\right)$-approximate respectively),
	$\epsilon v_{max}$-BIC, and IR mechanism for social welfare maximization.
	The running time is polynomial in the input size and $1 / \epsilon$. 
\end{corollary}

\paragraph{Algorithm.}

We will consider a variant of the LP-based algorithms by Feige 
\cite{feige2006maximizing} and use the reduction for social welfare to
convert it into an IR and $\epsilon v_{max}$-BIC mechanism. More concretely,
we will consider the Bayesian version of the standard 
social welfare maximization linear program (CA):
\begin{align*}
	\text{Maximize} ~ ~ \sum_i \sum_t \sum_S & ~ f_i(t) \, v^t_i(S) \,
	x_{i, t, S} & & \text{s.t.} \\
	\sum_i \sum_t \sum_{S : j \in S} f_i(t) \,x_{i, t, S} & ~ \leq ~ 1 & &
	\forall j \\
	\sum_S x_{i, t, S} & ~ \leq ~ 1 & & \forall i, t \\
	x_{i, t, S} & ~ \geq ~ 0 & & \forall i, t, S
\end{align*}

In this LP, $x_{i, t, S}$ denote the probability that agent $i$ is allocated 
with a subset of items $S$ conditioned on its valuation is $v^t_i$.
This LP can be solved in polynomial time by the standard primal dual 
technique via demand queries. See \cite{dobzinski2005approximation} for more
details. We let $LP^*$ denote the optimal value of this LP.
Moreover, for any basic feasible optimal solution of the above LP,
there are at most $nm\ell$ non-zero entries since there are only $nm\ell$
non-trivial constraints. Hence, we have the following lemma:

\begin{lemma}
	In poly-time we can find an optimal solution $\bm{x}^*$ to (CA) with at 
	most $nm\ell$ non-zero entries.
\end{lemma}

Next, we will filter this solution $\bm{x}^*$ by removing insignificant 
entries. We let $\hat{x}_{i, t, S} = x^*_{i, t, S} < \epsilon / n m \ell$.
Note that $LP^* \geq f_i(t) v^t_i(S)$ for any $i$, $t$, and $S$ since always
allocating subset $S$ to agent $i$ is a feasible allocation. We get that
$\hat{\bm{x}}$ is a feasible solution to (CA) with value at least 
$(1 - \epsilon) LP^*$.

Then, we will use the rounding algorithms by Feige \cite{feige2006maximizing}
to get a $\frac{1}{2}$-rounding
for sub-additive agents and a $\left(1-\frac{1}{e}\right)$-rounding for
fractionally sub-additive agents: 
\begin{enumerate}
	\item Allocate a tentative subset of items
		$\widetilde{S}_i$ to each agent $i$, $1 \leq i \leq n$, according
		to the reported valuation $v'_i = v^t_i$ and 
		$\hat{x}_{i, t, \widetilde{S}_i}$.
	\item Resolve conflicts properly by choosing $S_i \subseteq 
		\widetilde{S}_i$ so that $\bm{S} = (S_1, \dots, S_n)$ is a feasible
		allocation.
\end{enumerate}

By extending Feige's original proof, we can show that there is a randomized
algorithm for choosing $\bm{S}$ such that for sub-additive agents, we have:
\begin{equation}
	\label{eq:suba}
	\e{\bm{v}_{-i}, \bm{S}}{v_i(S_i)} \geq \frac{1}{2} \, v_i(\widetilde{S}_i)
	\enspace.
\end{equation}

And for fractionally sub-additive agents, we have:
\begin{equation}
	\label{eq:fracsuba}
	\e{\bm{v}_{-i}, \bm{S}}{v_i(S_i)} \geq 
	\left(1 - \frac{1}{e}\right) \, v_i(\widetilde{S}_i)
	\enspace.
\end{equation}

We will omit the proof in this extended abstract.
We denote the above algorithm as $\mathcal{A}$. Then, $\mathcal{A}$ provides 
$\left(\frac{1}{2} - \epsilon\right)$-approximation for sub-additive agents 
and $\left(1 - \frac{1}{e} - \epsilon\right)$-approximation for 
fractionally sub-additive agents.

\paragraph{Estimating values.}

By Theorem \ref{thm:welfarereduction} and \ref{thm:otherreduction}, we only
need to show how to estimate the values $\bm{w}_i$, $1 \leq i \leq n$, for
the induced assignment problem of agent $i$ efficiently. Further, by
Lemma \ref{lemma:relative}, we can efficiently estimate the values 
$\bm{w}_i = \{w^{st}_i\}_{1 \leq s, t \leq \ell}$, $1 \leq i \leq n$, if
the following lemma holds.

\begin{lemma}
	For any $1 \leq i \leq n$, and any $1 \leq s, t \leq \ell$,
	$$\frac{\sd{\bm{v}_{-i}, \bm{S}\sim\mathcal{A}(v^t_i, \bm{v}_{-i})}{v^s_i(S_i)}}
	{\e{\bm{v}_{-i}, \bm{S}\sim\mathcal{A}(v^t_i, \bm{v}_{-i})}{v^s_i(S_i)}}
	\leq \sqrt{\frac{4nm\ell}{\epsilon}} \enspace.$$
\end{lemma}

\begin{proof}
	By inequalities \eqref{eq:suba} and \eqref{eq:fracsuba}, we get that
	conditioned on $\widetilde{S}_i$ being chosen as the tentative set,
	\begin{equation*}
		{\e{\bm{v}_{-i}, \bm{S}\sim\mathcal{A}(v^t_i, \bm{v}_{-i})}
		{v^s_i(S_i) : \widetilde{S}_i}}
		\geq \frac{1}{2} \, v^s_i\left(\widetilde{S}_i\right) \enspace.
	\end{equation*}

	We also have that	
	$$\sd{\bm{v}_{-i}, \bm{S}\sim\mathcal{A}(v^t_i, \bm{v}_{-i})}
	{v^s_i(S_i) : \widetilde{S}_i} \leq \max \left\{v^s_i(S_i) : 
	\widetilde{S}_i\right\} \leq v^s_i(\widetilde{S}_i) \enspace.$$

	Hence, 
	\begin{align*}
		\sd{\bm{v}_{-i}, \bm{S}\sim\mathcal{A}(v^t_i, \bm{v}_{-i})}
		{v^s_i(S_i)}^2
		= \, & \sum_{\widetilde{S}_i}
		\hat{x}_{i, t, \widetilde{S}_i} \,
		\sd{\bm{v}_{-i}, \bm{S}\sim\mathcal{A}(v^t_i, \bm{v}_{-i})}
		{v^s_i(S_i) : \widetilde{S}_i}^2 \\
		\leq \, & \sum_{\widetilde{S}_i}
		\hat{x}_{i, t, \widetilde{S}_i} \, v^s_i(\widetilde{S}_i)^2 \\
		\leq \, & \frac{1}{\min \left\{\hat{x}_{i, t, \widetilde{S}_i} > 0
		\right\}} \, \left(\sum_{i, t, \widetilde{S}_i}
		\hat{x}_{i, t, \widetilde{S}_i} \, v^s_i(\widetilde{S}_i)\right)^2 \\
		\leq \, & \frac{nm\ell}{\epsilon} \, \left(\sum_{\widetilde{S}_i}
		\hat{x}_{i, t, \widetilde{S}_i} \,
		\e{\bm{v}_{-i}, \bm{S}\sim\mathcal{A}(v^t_i, \bm{v}_{-i})}
		{v^s_i(S_i) : \widetilde{S}_i}\right)^2 \\
		\leq \, & \frac{4nm\ell}{\epsilon} \, 
		\e{\bm{v}_{-i}, \bm{S}\sim\mathcal{A}(v^t_i, \bm{v}_{-i})}{v^s_i(S_i)}^2 
		\enspace.
	\end{align*}
\end{proof}

\bibliographystyle{alpha}
\bibliography{auctionbiblio}

\end{document}